\documentclass[11pt]{article}

\usepackage[a4paper,margin=1in]{geometry}
\usepackage{amsmath,amssymb,amsthm,mathtools,mathrsfs}
\usepackage{microtype}
\usepackage{xcolor}
\usepackage{booktabs}
\usepackage{float}
\usepackage{algorithm}
\usepackage{algpseudocode}
\usepackage{graphicx}
\usepackage[numbers,sort&compress]{natbib}
\usepackage[
  colorlinks=true,
  linkcolor=blue!55!black,
  citecolor=blue!55!black,
  urlcolor=blue!55!black
]{hyperref}

\newtheorem{theorem}{Theorem}[section]

\newtheorem{lemma}[theorem]{Lemma}

\theoremstyle{definition}

\newtheorem{remark}[theorem]{Remark}

\newcommand{\OPT}{\mathrm{OPT}}
\newcommand{\ADV}{\mathrm{ADV}}
\newcommand{\E}{\mathbb{E}}
\newcommand{\Prb}{\mathbb{P}}
\newcommand{\eps}{\varepsilon}
\newcommand{\cD}{\mathcal{D}}
\def\cost{\texttt{cost}}

\title{Online Line Aggregation with Deadlines:\\
Randomized Guarantees and Learning-Augmented Tradeoffs}
\author{Tianhang Lu\\
\small School of Mathematics Science, Ocean University of China\\
\small Qingdao 266100, Shandong, China\\
\small \texttt{lutianhang@stu.ouc.edu.cn}}
\date{}

\begin{document}

\maketitle

\begin{abstract}
We study online line aggregation with deadlines, where requests arrive over time on the positive half-line and a service at location $y$ clears all pending requests in $[0,y]$ at cost $y$.
In the classical adversarial setting, we propose an $e$-competitive randomized algorithm against an oblivious adversary and prove a matching lower bound.
Thus $e$ is the optimal randomized competitive ratio.
We then consider advice in the form of an offline feasible solution.
For every confidence parameter $\lambda\in(0,1]$, our deterministic learning-augmented algorithm is $(1+3/\lambda)$-robust and $(1+3\lambda)$-consistent.
We also propose a randomized learning-augmented algorithm that is $(e+e/\lambda)$-robust and $(e-1+\lambda)$-consistent against an oblivious adversary.
For the offline problem, we present a polynomial-time dynamic programming algorithm.
Numerical experiments complement the worst-case analysis: accurate advice lowers service costs, while both learning-augmented algorithms remain stable as the advice becomes increasingly noisy.

\end{abstract}

\section{Introduction}
\label{sec:intro}
Aggregation is a basic mechanism for exploiting economies of scale: several requests can be completed by one shared service rather than by many separate actions.
This principle appears in coordinated replenishment~\cite{askoy1988multi,goyal1989joint,joneja1990joint,khouja2008review}, lot sizing and logistics~\cite{brahimi2006single,jans2008modeling,quadt2008capacitated,karimi2003capacitated}, and communication systems~\cite{yuan2003synchronization}.
An online policy must decide not only when to serve but also how far a service should extend.
Extending a service can clear additional requests at little marginal cost, but using that capacity too early may forfeit a better aggregation opportunity later.

We study this decision problem on the positive half-line.
A request is specified by a release time, a location, and a deadline.
A service at location $y$ costs $y$ and clears every pending request in $[0,y]$.
Every request must be served between its release time and deadline, and the objective is to minimize the total service cost.
We call this problem \emph{line aggregation with deadlines} (LAD).
The line structure is simple enough to expose the central online decision---how far to extend a mandatory service---while retaining nontrivial interactions among requests released at different times.

LAD is the path specialization of the multi-level aggregation problem with deadlines (MLAP-D).
In MLAP-D, requests arrive at the vertices of a rooted weighted tree and a service is a rooted subtree that clears all pending requests it contains.
Bienkowski et al.~\cite{bienkowski2020mlap} developed an $O(D^22^D)$-competitive algorithm for trees of fixed depth $D$, and Buchbinder et al.~\cite{buchbinder2017depth} later obtained an $O(D)$-competitive algorithm.
Azar and Touitou~\cite{azar2019framework} developed a general framework for aggregation and related metric problems with delay or deadlines.
More recent work gives depth- and structure-parameterized guarantees~\cite{mcmahan2021dcompetitive,turoczy2025memory}.
Mari et al.~\cite{mari2024online} study additive-delay multi-level
aggregation under Poisson arrivals and obtain a constant ratio of
expectations, providing a complementary stochastic-input perspective.
For paths, the optimal deterministic competitive ratio for MLAP-D is $4$~\cite{bienkowski2021new}.

This tight deterministic result leaves open how much an oblivious adversary can gain against a randomized policy.
The deterministic lower bound does not apply directly to randomized algorithms, so it remains unclear how much randomization can improve worst-case performance.
Answering this question requires both an improved randomized upper bound and a matching lower bound for randomized algorithms.

Advice provides a second way to improve online decisions.
Learning-augmented algorithms use possibly inaccurate advice while seeking both \emph{consistency}, which measures performance when the advice is reliable, and \emph{robustness}, which protects against poor advice~\cite{purohit2018improving,lykouris2021competitive,bamas2020primaldual,mitzenmacher2022algorithms}.
For aggregation with deadlines, Dinitz, Fineman, and Umboh~\cite{dinitz2025jrpdeadlines} study predicted deadlines in a nonclairvoyant joint-replenishment model.

Finally, the offline structure of LAD is of independent interest.
For line aggregation with linear waiting costs, an exact polynomial-time dynamic program is known~\cite{bienkowski2013chain}.
For general offline MLAP-D, the problem is APX-hard even on trees of depth two~\cite{bienkowski2015approximation}.
For arbitrary trees, a polynomial-time $2$-approximation algorithm is known~\cite{bienkowski2021new}.

\begin{table}[t]
    \centering
    \caption{Overview of the main theoretical results.}
    \label{tab:results-overview}
    \small
    \setlength{\tabcolsep}{5pt}
    \begin{tabular}{@{}p{0.22\linewidth}p{0.30\linewidth}p{0.38\linewidth}@{}}
        \toprule
        Setting & Method & Main guarantee \\
        \midrule
        No advice
            & Shifted grid
            & $e$-competitive, and optimal \\
        \shortstack[l]{Advice\\(deterministic)}
            & \shortstack[l]{Advice-guided\\doubling}
            & $R=1+3/\lambda$, $C=1+3\lambda$ \\
        \shortstack[l]{Advice\\(randomized)}
            & \shortstack[l]{Advice-guided\\Shifted grid}
            & $R=e+e/\lambda$, $C=e-1+\lambda$ \\
        \shortstack[l]{Complete\\information}
            & Dynamic programming
            & $O(n^3)$ time, $O(n^2)$ space \\
        \bottomrule
    \end{tabular}
\end{table}

In this paper, we consider LAD under four different settings.
We first propose an $e$-competitive algorithm and prove a matching lower bound.
In the learning-augmented setting, for every $\lambda\in(0,1]$, we give a deterministic algorithm that is $(1+3/\lambda)$-robust and $(1+3\lambda)$-consistent, together with a randomized algorithm that is $(e+e/\lambda)$-robust and $(e-1+\lambda)$-consistent against an oblivious adversary.
We also develop an offline dynamic programming algorithm with running time $O(n^3)$ and space $O(n^2)$, where $n$ is the number of requests.
Finally, numerical experiments quantify how advice quality and the confidence parameter affect empirical performance and illustrate the complementary behavior of the deterministic and randomized learning-augmented algorithms.
Table~\ref{tab:results-overview} summarizes our main theoretical guarantees.
Here, $R$ and $C$ denote robustness and consistency, respectively.

The rest of this paper is organized as follows: Section~\ref{sec:preliminaries} defines the model and advice framework.
Section~\ref{sec:randomized} then establishes the tight randomized competitive ratio.
Building on these results, Sections~\ref{sec:learning-augmented-deterministic} and~\ref{sec:learning-augmented-randomized} present the deterministic and randomized learning-augmented trade-offs, respectively.
Section~\ref{sec:offlineoptimal} subsequently gives the offline algorithm.
Section~\ref{sec:experiments} reports the numerical experiments.
Finally, Section~\ref{sec:conclusion} summarizes the main findings and concludes the paper.

\section{Preliminaries}
\label{sec:preliminaries}
\subsection{Line Aggregation with Deadlines}
\label{subsec:model}

An instance $I$ consists of a finite sequence of requests on the half-line $\mathbb{R}_{\ge 0}$ rooted at the origin.
Each request is a triple $r=(\tau,x,d)$, where $\tau\ge 0$ is its release time, $x>0$ is its location, and $d\ge \tau$ is its deadline.
A request remains \emph{pending} until it is served.
When releases and deadlines occur at the same time, newly released requests are revealed before the deadline is processed.

A solution is represented by a set of time-location pairs.
We write this set as $X=\{(t_1,x_1),\ldots,(t_n,x_n)\}$.
Here, $(t_i,x_i)$ denotes a service performed at location $x_i$ at time $t_i$, and the cost of this service is $x_i$.
The solution $X$ is feasible if, for every request $r=(\tau,x,d)$, there exists a pair $(t_i,x_i)\in X$ such that $\tau\leq t_i\leq d$ and $x_i\geq x$.
The cost of $X$ is the sum of the costs of all services in $X$.
For any algorithm $\mathcal{A}$ that produces a feasible solution, let $\cost(\mathcal{A})$ denote the cost of its output.

\subsection{Advice.}
\label{subsec:advice}
The algorithm receives advice online in the form of suggested service locations, which may be provided at arbitrary times.
At any time $t$, the advice may suggest a service at location $a_t>0$.
For a request $r=(\tau,x,d)$, the advice \emph{covers} $r$ at time $d\geq t\geq\tau$ if $a_t\geq x$.
We denote the first such time by $t_{\text{cov}}(r)$ and call it the \emph{first coverage time} of $r$.
We call the advice \emph{feasible} if it eventually covers every request, and assume this condition throughout this paper.

\subsection{Performance Guarantees}
\label{subsec:performance}

In the learning-augmented setting, the algorithm receives advice in the form of an offline feasible solution.
For each instance $I$, let $\OPT$ and $\ADV$ denote an optimal offline solution and the solution prescribed by the advice, respectively.
A randomized online algorithm $\mathcal{A}$ is called \textit{$r$-competitive}, or \textit{$r$-robust}, if $\mathbb{E}[\cost(\mathcal{A})]\le r\cdot\cost(\OPT)$ for any instance $I$.
Moreover, $\mathcal{A}$ is called \textit{$c$-consistent} if $\mathbb{E}[\cost(\mathcal{A})]\le c\cdot\cost(\ADV)$ for any instance $I$ and any feasible advice $\ADV$.
For deterministic algorithms, the expectations are omitted.
Robustness measures protection against any feasible advice, whereas consistency measures performance when the advice is reliable.

\subsection{The \textmd{\textsc{Doubling}} Algorithm}
\label{subsec:doubling}
For each $i\in\mathbb{Z}$, define the service level $b_i=2^i$.
When one or more pending requests reach their deadlines, let $x$ be the largest location among these requests and define $H(x)=\min\{i\in\mathbb{Z}:b_i\geq x\}$.
The \textsc{Doubling} algorithm performs a service at location $b_{H(x)}$.
Since $x\leq b_{H(x)}<2x$, this service serves every request reaching its deadline at that time.
\begin{theorem}\cite{bienkowski2020mlap}.
\label{thm:doubling-benchmark}
The \textup{\textsc{Doubling}} algorithm is $4$-competitive for LAD.
\end{theorem}

\section{Randomized Algorithms for LAD}
\label{sec:randomized}
In this section, we first present a randomized online algorithm for LAD.
We then prove a matching lower bound for all randomized online algorithms, which shows that the competitive ratio achieved by our algorithm is optimal.

\subsection{A Randomized Algorithm for LAD}
\label{subsec:randomized-algorithm}

Our randomized online algorithm retains the main rule of \textsc{Doubling} but uses randomly shifted exponential service levels.
At the beginning, it samples $U$ uniformly from $[0,1)$.
For every $i\in\mathbb{Z}$, set $b_i=e^{i+U}$, and $H(x)=\min\{i\in\mathbb{Z}:b_i\geq x\}$.
Whenever one or more pending requests reach their deadlines, let $x$ denote the largest location among these requests.
The algorithm performs a service at location $b_{H(x)}$, as summarized in Algorithm~\ref{alg:randomized-lad}.
Since $b_{H(x)}\geq x$, this service serves every request reaching its deadline at that time.
Thus, Algorithm~\ref{alg:randomized-lad} always produces a feasible solution.

\begin{algorithm}[h]
\caption{An online randomized algorithm for LAD}
\label{alg:randomized-lad}
\begin{algorithmic}[1]
\State Sample $U$ uniformly from $[0,1)$.
\State Set $b_i\gets e^{i+U}$ for every $i\in\mathbb{Z}$.
\While{pending requests exist}
    \If{one or more pending requests reach their deadlines}
        \State Let $x$ denote the largest location among these requests, and perform a service at location $b_{H(x)}$.
    \EndIf
\EndWhile
\end{algorithmic}
\end{algorithm}

We next analyze the competitive ratio of the proposed algorithm.
\begin{theorem}
\label{thm:randomized-upper-bound}
Algorithm \ref{alg:randomized-lad} is $e$-competitive for LAD against an oblivious adversary.
\end{theorem}

\begin{proof}
Fix an instance $I$ and an optimal offline solution $\OPT$.
Denote Algorithm~\ref{alg:randomized-lad} by $\mathcal{A}$.
Consider a service performed by $\mathcal{A}$ at time $t$, and choose a request at location $x$ that reaches its deadline at $t$ and determines the service location $b_{H(x)}$.
Charge the cost $b_{H(x)}$ incurred by this service to the first service in $\OPT$ that serves request $x$.
Thus, each service performed by $\mathcal{A}$ is charged exactly once.

Fix a service $(s,A)\in\OPT$, whose cost is $A$, and consider all services of $\mathcal{A}$ charged to it in chronological order.
Their indices are strictly increasing.
To see this, suppose that a service at time $t$ uses $b_i$ and a later service at time $t'>t$ uses $b_j$, where both are charged to $(s,A)$ and $j\leq i$.
Let $r'$ denote the request that determines the later service.
Since $(s,A)$ serves the request determining the earlier service, whose deadline is $t$, we have $s\leq t$.
Since $(s,A)$ also serves $r'$, the release time of $r'$ is at most $s$ and its location is at most $A$.
Hence $r'$ has already been released by time $t$.
Moreover, since $r'$ determines the service at $t'$, its location is at most $b_j\leq b_i$.
Therefore, the service at location $b_i$ performed at time $t$ would have served $r'$, contradicting that $r'$ remains pending until $t'$.
Thus, the indices of the services charged to $(s,A)$ are strictly increasing.

Every request assigned to $(s,A)$ has location at most $A$, so every service charged to $(s,A)$ has index at most $H(A)$.
Since no index occurs more than once, the total charge to $(s,A)$, for any fixed $U$, is at most
$$
\sum_{i\leq H(A)}b_i=\sum_{i\leq H(A)}e^{i+U}=\frac{e}{e-1}b_{H(A)}.
$$
It remains to take the expectation over $U$.
Except on a probability-zero event, $H(A)=\lceil\ln A-U\rceil$.
Since $U$ is uniform on $[0,1)$, the logarithmic excess $H(A)+U-\ln A$ is also uniform on $[0,1)$.
Consequently,
$$
\E[b_{H(A)}]=A\int_0^1 e^z\,dz=(e-1)A.
$$
The expected total charge to $(s,A)$ is therefore at most
$$
\frac{e}{e-1}\E[b_{H(A)}]=eA.
$$
Summing over all services in $\OPT$ and using linearity of expectation gives
$$
\E[\cost(\mathcal{A})]\leq\sum_{(s,A)\in\OPT}eA=e\,\cost(\OPT).
$$
Since this inequality holds for every instance fixed independently of the random choice $U$, $\mathcal{A}$ is $e$-competitive against an oblivious adversary.
\end{proof}

\begin{remark}
The choice of $e$ in Algorithm~\ref{alg:randomized-lad} is optimal.
For any $\beta>1$, the base $e$ in the service levels can be replaced by $\beta$.
The same analysis shows that the resulting randomized algorithm is $(\beta/\ln\beta)$-competitive. 
This ratio is minimized at $\beta=e$, giving a competitive ratio of $e$.
\end{remark}

\subsection{A Matching Randomized Lower Bound}
\label{subsec:randomized-lower-bound}
We now prove a matching lower bound for every randomized online algorithm for LAD.
\begin{theorem}\label{thm:lower-bound}
Every randomized online algorithm for LAD has competitive ratio at least $e$ against an oblivious adversary.
\end{theorem}

By Yao's principle~\cite{yao1977probabilistic}, it suffices to construct a distribution under which every deterministic algorithm has a competitive ratio approaching $e$.
The hard instance $I$ is constructed as follows.
Fix $\eps\in(0,1]$ and integers $N\ge 0$ and $M\ge 1$.
Define $x_i=e^{i\eps}$, $i=0,\ldots,N$ and independently sample $K_1,\ldots,K_M$ from $\{0,\ldots,N\}$ whose tail probabilities satisfy $\Prb[K_p\ge i]=e^{-i\eps}$, $i=0,1,\ldots,N$.
For $p=1,\ldots,M$, let $\Delta_p$ denote $\bigl(4(N+2)\bigr)^{-(p-1)}$.
Set $\sigma_1=0$ and
\begin{equation}\label{eq:lower-release-times}
\sigma_{p+1} = \sigma_p+\left(K_p+\frac32\right)\Delta_p,\qquad p=1,\ldots,M.
\end{equation}
Thus, $I$ consists of $M$ consecutive time intervals $[\sigma_p,\sigma_{p+1}]$, with $p$-th interval having length $(K_p+3/2)\Delta_p$.
At time $\sigma_p$, release the requests
\begin{equation}\label{eq:lower-requests}
r_{p,i}=(\sigma_p,x_i,d_{p,i}), \qquad i=0,\ldots,N,
\end{equation}
where $d_{p,i} = \sigma_p+(i+1)\Delta_p$, $i=0,\ldots,N$.
The instance also contains the request $q=(\sigma_{M+1},x_0,\sigma_{M+1})$.
This request makes the value of $K_M$, which determines $\sigma_{M+1}$, part of the instance.
It can be served at its deadline since, by the convention in Section~\ref{subsec:model}, releases are revealed before deadlines at the same time.

We record the separation between the deadlines and the later release times.
For every $p=1,\ldots,M$,
\begin{align*}
0 &\le \sigma_{M+1}-\sigma_{p+1} = \sum_{j=p+1}^{M}\left(K_j+\frac32\right)\Delta_j \notag\\
&\le \left(N+\frac32\right) \sum_{j=p+1}^{\infty}\Delta_j = \frac{N+\frac32}{4(N+2)-1}\Delta_p < \frac14\Delta_p.
\end{align*}
Together with \eqref{eq:lower-release-times}, this implies
\begin{equation}\label{eq:lower-separation}
d_{p,K_p} < \sigma_{p+1} \le \sigma_{M+1}, \qquad \sigma_{M+1}<d_{p,K_p+1}\quad\text{if }K_p<N.
\end{equation}
Consequently, $r_{p,0},\ldots,r_{p,K_p}$ reach their deadlines before $\sigma_{p+1}$, whereas every request $r_{p,i}$ with $i>K_p$ satisfies
$$
d_{p,i}\ge d_{p,K_p+1}>\sigma_{M+1}.
$$
Equations \eqref{eq:lower-requests} define the distribution $\cD_{M,N,\eps}$.
All $K_p$ are sampled before the online algorithm starts, so $\cD_{M,N,\eps}$ is chosen by an oblivious adversary.
We next establish a lower bound on the expected cost incurred by any deterministic algorithm.

\begin{lemma}\label{lem:online-lower-bound}
For every deterministic online algorithm $\mathcal A$,
\begin{equation}\label{eq:online-lower-bound}
\E_{I\sim\cD_{M,N,\eps}}[\cost(\mathcal A)] \ge M e^{1-\eps}(N+1)\eps.
\end{equation}
\end{lemma}

\begin{proof}
For $p=1,\ldots,M$, let $C_p$ denote the sum of the costs of all services that $\mathcal A$ performs during $[\sigma_p,\sigma_{p+1})$.
These time intervals are pairwise disjoint, and hence
\begin{equation}\label{eq:lower-disjoint-cost}
\cost(\mathcal A)\ge\sum_{p=1}^M C_p.
\end{equation}

Fix $p$ and condition on the complete history before $\sigma_p$.
This fixes the state of $\mathcal A$ and all pending requests.
Every pending request released before $\sigma_p$ has a deadline after $\sigma_{M+1}$, and hence after $\sigma_{p+1}$, by \eqref{eq:lower-separation}.
Indeed, a request $r_{j,i}$ with $j<p$ and $i\le K_j$ reaches its deadline before $\sigma_{j+1}\le\sigma_p$ and therefore cannot still be pending.

Consider the execution corresponding to $K_p=N$, and order all services made during $[\sigma_p,\sigma_{p+1})$ chronologically, breaking ties by execution order.
We select some of these services and use them to partition the indices $\{0,\ldots,N\}$.
Set $a_1=0$.
Given $a_s$, define the $s$-th selected service as the first service at a location $\ell_s\ge x_{a_s}$, and define
$$
b_s=\max\{i\in\{0,\ldots,N\}:x_i\le\ell_s\}.
$$
If $b_s=N$, stop; otherwise, set $a_{s+1}=b_s+1$ and continue.
The selected service exists and occurs no later than $d_{p,a_s}$, since the request $r_{p,a_s}$ must be served by that deadline.
Therefore, the intervals of indices $[a_s,b_s]$ form a partition of $\{0,\ldots,N\}$.

If $K_p\ge a_s$, then
$$
\sigma_{p+1}=\sigma_p+\left(K_p+\frac32\right)\Delta_p>d_{p,a_s}.
$$
Up to the selected service, the requests revealed to $\mathcal A$ are thus identical to those in the execution with $K_p=N$.
Since $\mathcal A$ is deterministic, it performs the same service at $\ell_s$ whenever $K_p\ge a_s$.
It follows that, conditional on the history before $\sigma_p$,
\begin{align*}
\E[C_p\mid\text{history before }\sigma_p] &\ge \sum_s \Prb[K_p\ge a_s]\ell_s \\
&\ge \sum_s e^{-a_s\eps}x_{b_s} = \sum_s e^{(b_s-a_s)\eps}.
\end{align*}
Let $n_s$ denote $b_s-a_s+1$.
Since $e^z\ge ez$ for $z>0$,
$$
e^{(b_s-a_s)\eps}=e^{-\eps}e^{n_s\eps}\ge e^{1-\eps}n_s\eps.
$$
The index intervals form a partition, so $\sum_s n_s=N+1$.
Therefore,
\begin{equation}\label{eq:lower-conditional-cost}
\E[C_p\mid\text{history before }\sigma_p] \ge e^{1-\eps}(N+1)\eps.
\end{equation}
Previously released pending requests may affect the services chosen by $\mathcal A$, but they do not invalidate the argument: none reaches its deadline before $\sigma_{p+1}$, and every service cost is included in $C_p$.

Taking expectations in \eqref{eq:lower-conditional-cost}, summing over $p$, and applying \eqref{eq:lower-disjoint-cost} proves \eqref{eq:online-lower-bound}.
\end{proof}

\begin{lemma}\label{lem:offline-upper-bound}
For the distribution $\cD_{M,N,\eps}$,
\begin{equation}\label{eq:offline-upper-bound}
\E_{I\sim\cD_{M,N,\eps}}[\cost(\OPT)] \le M\bigl(1+N(1-e^{-\eps})\bigr)+e^{N\eps}.
\end{equation}
\end{lemma}

\begin{proof}
Consider the following offline solution.
For every $p=1,\ldots,M$, it performs a service at location $x_{K_p}$ at time $d_{p,0}$.
All requests in \eqref{eq:lower-requests} have been released by that time, and this service serves $r_{p,0},\ldots,r_{p,K_p}$ no later than their deadlines.
At time $\sigma_{M+1}$, the solution performs one service at location $x_N$.
By \eqref{eq:lower-separation}, this service occurs before the deadlines of all remaining requests $r_{p,i}$ with $i>K_p$.
It also serves $q$.
The solution is therefore feasible and has cost
$$
\sum_{p=1}^M x_{K_p}+x_N.
$$
For each $p$, the tail-sum identity gives
\begin{align*}
\E[x_{K_p}] &= x_0+\sum_{i=1}^N(x_i-x_{i-1})\Prb[K_p\ge i] \\
&=1+\sum_{i=1}^N \left(e^{i\eps}-e^{(i-1)\eps}\right)e^{-i\eps} =1+N(1-e^{-\eps}).
\end{align*}
Since $x_N=e^{N\eps}$, the expected cost of this feasible solution is the right-hand side of \eqref{eq:offline-upper-bound}.
Since $\OPT$ has no larger cost than this feasible solution, \eqref{eq:offline-upper-bound} follows.
\end{proof}

\begin{proof}[Proof of Theorem~\ref{thm:lower-bound}]
Combining Lemmas~\ref{lem:online-lower-bound} and \ref{lem:offline-upper-bound}, every deterministic online algorithm $\mathcal A$ satisfies
\begin{equation}\label{eq:lower-finite-ratio}
\frac{\E_{I\sim\cD_{M,N,\eps}}[\cost(\mathcal A)]} {\E_{I\sim\cD_{M,N,\eps}}[\cost(\OPT)]} \ge \frac{M e^{1-\eps}(N+1)\eps} {M\bigl(1+N(1-e^{-\eps})\bigr)+e^{N\eps}}.
\end{equation}
For fixed $N$ and $\eps$, letting $M\to\infty$ removes the one-time cost $e^{N\eps}$ from the ratio.
Letting $N\to\infty$ next gives
$$
e^{1-\eps}\frac{\eps}{1-e^{-\eps}}.
$$
Finally, $1-e^{-\eps}\sim\eps$ as $\eps\to 0$, and hence
$$
\lim_{\eps\to 0}\lim_{N\to\infty}\lim_{M\to\infty}\frac{M e^{1-\eps}(N+1)\eps}{M\bigl(1+N(1-e^{-\eps})\bigr)+e^{N\eps}}=e.
$$
Thus, for every $\delta>0$, we may first choose $\eps>0$ sufficiently small, then $N$ sufficiently large, and finally $M$ sufficiently large so that the right-hand side of \eqref{eq:lower-finite-ratio} is at least $e-\delta$.
The resulting $\cD_{M,N,\eps}$ is a finite-support distribution that is independent of $\mathcal A$.

To apply Yao's principle, fix any randomized online algorithm $\mathcal A$ and condition on its random choices.
Each outcome gives a deterministic online algorithm, to which the preceding distributional inequality applies.
Averaging first over the random choices of $\mathcal A$ and then over $I\sim\cD_{M,N,\eps}$ preserves this inequality.
Since $\cD_{M,N,\eps}$ has finite support, some instance in its support satisfies
$$
\E[\cost(\mathcal A)]\ge(e-\delta)\cost(\OPT),
$$
where the expectation is over the random choices of $\mathcal A$.
Hence no randomized online algorithm against an oblivious adversary can have competitive ratio strictly smaller than $e$.
Together with Theorem~\ref{thm:randomized-upper-bound}, this establishes $e$ as the optimal randomized competitive ratio.
\end{proof}

\section{Learning-Augmented Deterministic Algorithms}
\label{sec:learning-augmented-deterministic}
We give a deterministic learning-augmented algorithm whose reliance on the advice is controlled by a confidence parameter $\lambda\in(0,1]$.
The algorithm performs a service only when a pending request reaches its deadline.
Suppose that $r=(\tau,x,d)$ is such a request.
By feasibility of the advice, $t_{\mathrm{cov}}(r)\leq d$, so the advised service location $a_{t_{\mathrm{cov}}(r)}$ is available when $d$ is processed.
If this location is at most $2x/\lambda$, the algorithm performs a service at location $\max\left\{a_{t_{\mathrm{cov}}(r)},(1+\lambda)x\right\}$.
Otherwise, it performs a service at location $2x$.
The complete procedure is stated in Algorithm~\ref{alg:pd-doubling}.
We next analyze the robustness and consistency of this algorithm.

\begin{algorithm}[h]
\caption{Learning-augmented deterministic algorithm for LAD}
\label{alg:pd-doubling}
\begin{algorithmic}[1]
\Require Confidence parameter $\lambda\in(0,1]$
\For{each release, advice, or deadline time $t$, in increasing order}
    \While{there exists a pending request $r=(\tau,x,d)$ with $d=t$}
        \If{$a_{t_{\mathrm{cov}}(r)}\leq 2x/\lambda$}
            \State $b\gets\max\{a_{t_{\mathrm{cov}}(r)},(1+\lambda)x\}$
        \Else
            \State $b\gets2x$
        \EndIf
        \State Perform a service at location $b$ at time $t$
    \EndWhile
\EndFor
\end{algorithmic}
\end{algorithm}

\begin{theorem}
\label{thm:pd-doubling-tradeoff}
For every $\lambda\in(0,1]$, Algorithm~\ref{alg:pd-doubling} is $(1+3/\lambda)$-robust and $(1+3\lambda)$-consistent.
\end{theorem}

\begin{proof}
Let $\mathcal{A}$ denote Algorithm~\ref{alg:pd-doubling}.
We first establish feasibility.
Whenever a request $r=(\tau,x,d)$ triggers a service, the service location is at least $(1+\lambda)x>x$.
Thus, $r$ is served by its deadline.
Under the event convention in Section~\ref{subsec:model}, every request with deadline $t$ is either served by an earlier service or triggers a service at time $t$.

We next prove robustness by charging the algorithm's services to $\OPT$.
Since the instance is finite, redundant services can be removed from $\OPT$, so we may assume that $\OPT$ is finite.
For each service at location $b$ triggered by a request $r=(\tau,x,d)$, choose one service $(s,y)\in\OPT$ satisfying $\tau\leq s\leq d$ and $y\geq x$, and charge $b$ to it.
Such a service exists because $\OPT$ is feasible.

Fix $(s,y)\in\OPT$, and list the algorithmic services charged to it in execution order.
Let $r_i=(\tau_i,x_i,d_i)$ be the request triggering the $i$th such service, and let $b_i$ be its service location.
For every $i<m$, the service $(s,y)$ covers both $r_i$ and $r_{i+1}$, and hence
$$
\tau_{i+1}\leq s\leq d_i.
$$
Therefore, $r_{i+1}$ has been released when the service triggered by $r_i$ is performed.
Because $r_{i+1}$ later triggers another service, it is not served by the service at location $b_i$, which implies
$$
x_{i+1}>b_i\geq(1+\lambda)x_i.
$$
Consequently, $x_j<x_m/(1+\lambda)^{m-j}$ for every $j<m$.
Moreover, $x_m\leq y$ because $(s,y)$ covers $r_m$, and $b_m\leq2x_m/\lambda\leq2y/\lambda$.
The geometric bound also gives $\sum_{i=2}^m x_i<(1+\lambda)x_m/\lambda$, where the sum is empty if $m=1$.
The total cost charged to $(s,y)$ is
$$
\sum_{i=1}^m b_i\leq\sum_{i=2}^m x_i+\frac{2y}{\lambda}<\left(\frac{1+\lambda}{\lambda}+\frac{2}{\lambda}\right)y
=\left(1+\frac{3}{\lambda}\right)y.
$$
Every service performed by the algorithm is charged to exactly one service in $\OPT$.
Summing over all services in $\OPT$ gives
$$
\cost(\mathcal{A})\leq\left(1+\frac{3}{\lambda}\right)\cost(\OPT),
$$
which proves robustness.

We now prove consistency by charging each algorithmic service to the first advice service that covers its triggering request.
Fix an advice service $(s,a_s)\in\ADV$, and list in execution order the algorithmic services whose triggering requests $r_i=(\tau_i,x_i,d_i)$ satisfy $t_{\mathrm{cov}}(r_i)=s$.
Let $b_i$ be the location of the $i$th such service.
For every $i$, the definition of first coverage gives $\tau_i\leq s\leq d_i$ and $x_i\leq a_s$.
For $i<m$, request $r_{i+1}$ has therefore been released when the service triggered by $r_i$ is performed.
Because $r_{i+1}$ later triggers a service, it is not served by the service at location $b_i$, so $x_{i+1}>b_i$.

If a nonfinal service in this group took the first branch of Algorithm~\ref{alg:pd-doubling}, its location would be at least $a_s$.
All requests whose first coverage time is $s$ have been released by time $s\leq d_i$ and have locations at most $a_s$.
Such a service would therefore serve every remaining request in the group, contradicting the existence of a later trigger.
Hence, at most one service in the group takes the first branch, and if it exists, it is the final service.

Every other service in the group is performed at location $2x_i$ and satisfies
$$
a_s>\frac{2x_i}{\lambda},
\qquad
b_i=2x_i<\lambda a_s.
$$
For two consecutive such services, $x_{i+1}>b_i$ implies $b_{i+1}=2x_{i+1}>2b_i$.
Thus, their service locations form a sequence that grows by a factor greater than $2$, and their total cost is less than $2\lambda a_s$.
If the final service takes the first branch, then $x_m\leq a_s$ and its location is at most
$$
\max\{a_s,(1+\lambda)x_m\}\leq(1+\lambda)a_s.
$$
The total cost charged to $(s,a_s)$ is therefore less than $(1+3\lambda)a_s$.

Every service performed by the algorithm is charged to exactly one advice service because the advice is feasible.
Summing over all services in $\ADV$ gives
$$
\cost(\mathcal{A})\leq(1+3\lambda)\cost(\ADV),
$$
which proves consistency.
\end{proof}

\section{Learning-Augmented Randomized Algorithms}
\label{sec:learning-augmented-randomized}
We now combine the advice with the randomly shifted service levels from Section~\ref{subsec:randomized-algorithm}.
We sample the random offset $U$ once and use the same value throughout the algorithm.
Suppose that a pending request $r=(\tau,x,d)$ reaches its deadline.
By feasibility of the advice, $a_{t_{\mathrm{cov}}(r)}$ is available when $d$ is processed.
If this advised service location is at most $ex/\lambda$, the algorithm performs a service at location $\max\left\{a_{t_{\mathrm{cov}}(r)},b_{H(x)}\right\}$.
Otherwise, it performs a service at location $b_{H(x)}$.
The complete procedure is stated in Algorithm~\ref{alg:learning-randomized}.
We first establish the following identity for expectations, which will be used later in the proof.

\begin{algorithm}[h]
\caption{Learning-augmented randomized algorithm for LAD}
\label{alg:learning-randomized}
\begin{algorithmic}[1]
\Require Confidence parameter $\lambda\in(0,1]$
\State Sample $U$ uniformly from $[0,1)$
\State Set $b_i\gets e^{i+U}$ for every $i\in\mathbb{Z}$
\For{each release, advice, or deadline time $t$, in increasing order}
    \While{there exists a pending request $r=(\tau,x,d)$ with $d=t$}
        \If{$a_{t_{\mathrm{cov}}(r)}\leq ex/\lambda$}
            \State $b\gets\max\{a_{t_{\mathrm{cov}}(r)},b_{H(x)}\}$
        \Else
            \State $b\gets b_{H(x)}$
        \EndIf
        \State Perform a service at location $b$ at time $t$
    \EndWhile
\EndFor
\end{algorithmic}
\end{algorithm}

\begin{lemma}
\label{lem:random-grid-identities}
For every fixed $z>0$,
$$
\mathbb{E}[b_{H(z)}]=(e-1)z,
\qquad
\mathbb{E}\left[\sum_{i\leq H(z)}b_i\right]=ez,
$$
and
$$
\mathbb{E}\left[\sum_{b_i<z}b_i\right]=z.
$$
\end{lemma}

\begin{proof}
Except on a probability-zero event, $\ln b_{H(z)}-\ln z$ is uniform on $[0,1)$.
Consequently,
$$
\mathbb{E}[b_{H(z)}]
=z\int_0^1 e^v\,dv
=(e-1)z.
$$
For every fixed offset $U$, we have
$$
\sum_{i\leq H(z)}b_i=\frac{e}{e-1}b_{H(z)}
$$
and, except on the same probability-zero event,
$$
\sum_{b_i<z}b_i=\frac{1}{e-1}b_{H(z)}.
$$
Taking expectations proves the remaining identities.
\end{proof}

\begin{theorem}
\label{thm:randomized-learning-tradeoff}
For every $\lambda\in(0,1]$, Algorithm~\ref{alg:learning-randomized} is $(e+e/\lambda)$-robust and $(e-1+\lambda)$-consistent against an oblivious adversary.
\end{theorem}

\begin{proof}
Let $\mathcal{A}$ denote Algorithm~\ref{alg:learning-randomized}.
We first establish feasibility.
Every service triggered by a request $r=(\tau,x,d)$ is performed at a location at least $b_{H(x)}\geq x$.
Thus, $r$ is served by its deadline.
Under the event convention in Section~\ref{subsec:model}, every request with deadline $t$ is either served by an earlier service or triggers a service at time $t$.

We next prove robustness by charging the algorithm's services to $\OPT$.
Since the instance is finite, redundant services can be removed from $\OPT$, so we may assume that $\OPT$ is finite.
For each service triggered by a request $r=(\tau,x,d)$, choose one service $(s,y)\in\OPT$ satisfying $\tau\leq s\leq d$ and $y\geq x$, and charge the algorithmic service to it.
Such a service exists because $\OPT$ is feasible.

Fix $(s,y)\in\OPT$, and list the algorithmic services charged to it in execution order.
Let $r_j=(\tau_j,x_j,d_j)$ be the request triggering the $j$th such service, and let $c_j$ be its service location.
For every $j<m$, the service $(s,y)$ covers both $r_j$ and $r_{j+1}$, so
$$
\tau_{j+1}\leq s\leq d_j.
$$
Therefore, $r_{j+1}$ has been released when the service triggered by $r_j$ is performed.
Because $r_{j+1}$ later triggers another service, it is not served by the service at location $c_j$, which implies
$$
x_{j+1}>c_j\geq b_{H(x_j)}.
$$
Hence the indices $H(x_1),\ldots,H(x_m)$ are strictly increasing.
Moreover, $x_j\leq y$ for every $j$, and every non-final service satisfies $c_j<x_{j+1}\leq b_{H(x_{j+1})}$.
It follows that
$$
\sum_{j=1}^{m-1}c_j
<\sum_{j=2}^m b_{H(x_j)}
\leq\sum_{i\leq H(y)}b_i.
$$

The final service satisfies $c_m\leq ey/\lambda$.
Indeed, if it takes the first branch, then both $a_{t_{\mathrm{cov}}(r_m)}\leq ex_m/\lambda$ and $b_{H(x_m)}<ex_m\leq ex_m/\lambda$.
If it takes the second branch, then $c_m=b_{H(x_m)}<ex_m\leq ex_m/\lambda$.
In both cases, $x_m\leq y$ gives the claimed bound.
Lemma~\ref{lem:random-grid-identities} therefore yields
$$
\mathbb{E}\left[\sum_{j=1}^m c_j\right]
\leq\left(e+\frac{e}{\lambda}\right)y.
$$
Every service performed by the algorithm is charged to exactly one service in $\OPT$.
Summing over all services in $\OPT$ gives
$$
\mathbb{E}[\cost(\mathcal{A})]
\leq\left(e+\frac{e}{\lambda}\right)\cost(\OPT),
$$
which proves robustness against an oblivious adversary.

We now prove consistency by charging each algorithmic service to the first advice service that covers its triggering request.
Fix an advice service $(s,a_s)\in\ADV$, and list in execution order the algorithmic services whose triggering requests $r_j=(\tau_j,x_j,d_j)$ satisfy $t_{\mathrm{cov}}(r_j)=s$.
Let $c_j$ be the location of the $j$th such service.
For every $j$, the definition of first coverage gives $\tau_j\leq s\leq d_j$ and $x_j\leq a_s$.
For $j<m$, request $r_{j+1}$ has therefore been released when the service triggered by $r_j$ is performed.
Because $r_{j+1}$ later triggers a service, it is not served by the service at location $c_j$, so $x_{j+1}>c_j$.

If a non-final service in this group took the first branch of Algorithm~\ref{alg:learning-randomized}, its location would be at least $a_s$.
All requests whose first coverage time is $s$ have been released by time $s\leq d_j$ and have locations at most $a_s$.
Such a service would therefore serve every remaining request in the group, contradicting the existence of a later trigger.
Hence, at most one service in the group takes the first branch, and if it exists, it is the final service.

Every other service in the group is performed at location $b_{H(x_j)}$ and satisfies
$$
a_s>\frac{ex_j}{\lambda},
\qquad
c_j=b_{H(x_j)}<ex_j<\lambda a_s.
$$
The separation $x_{j+1}>c_j$ shows that these service locations are distinct increasing grid levels.
Their total is therefore at most the sum of all grid levels strictly below $\lambda a_s$.
By Lemma~\ref{lem:random-grid-identities}, their expected total cost is at most $\lambda a_s$.

If the final service takes the first branch, then its location is at most $b_{H(a_s)}$.
Indeed, if $b_{H(x_m)}\leq a_s$, its location is $a_s\leq b_{H(a_s)}$.
Otherwise, no grid level lies between $x_m$ and $a_s$, so $H(x_m)=H(a_s)$.
Lemma~\ref{lem:random-grid-identities} therefore bounds the expected cost of this final service by $(e-1)a_s$.
The same bound holds trivially when no service takes the first branch.
Thus, the expected total cost charged to $(s,a_s)$ is at most $(e-1+\lambda)a_s$.

Every service performed by the algorithm is charged to exactly one advice service because the advice is feasible.
Summing over all services in $\ADV$ gives
$$
\mathbb{E}[\cost(\mathcal{A})]
\leq(e-1+\lambda)\cost(\ADV),
$$
which proves consistency.
\end{proof}

\section{Polynomial-Time Offline Algorithm}
\label{sec:offlineoptimal}
We show that offline LAD is exactly solvable in polynomial time in this section.
We first establish the following properties of optimal solutions to offline LAD.
These properties guide the construction of the states in our dynamic program.
Throughout this section, write $a_r$, $x_r$, and $d_r$ for the release time, location, and deadline of request $r$, respectively.

\begin{lemma}
\label{lem:canonical}
For every finite instance $J$, there exists an optimal solution $S^\star$ and an assignment of each request to one service that covers it such that every service $s=(t_s,y_s)\in S^\star$ satisfies
\begin{enumerate}
 \item at least one request is assigned to $s$;
 \item $y_s$ is the largest location among the requests assigned to $s$;
 \item $t_s$ is the smallest deadline among the requests assigned to $s$.
\end{enumerate}
Moreover, $S^\star$ contains at most one service at each deadline.
In particular, every service in $S^\star$ occurs at the deadline of a request in $J$, and $\lvert S^\star\rvert \leq \lvert J\rvert$.
\end{lemma}

\begin{proof}
Begin with any feasible solution and assign each request to one service that covers it.
Delete every service with no assigned request.
For a remaining service $s$, let $R(s)$ be its nonempty set of assigned requests.
Replace its location by $y'_s=\max_{r\in R(s)}x_r$.
The original service covers every request in $R(s)$, so $y'_s\le y_s$.
This replacement preserves feasibility and does not increase the cost.

Next, move the service to $t'_s=\min_{r\in R(s)}d_r$.
For every $r\in R(s)$, feasibility of the original service gives $a_r\le t_s\le d_r$.
Consequently, $t_s\le t'_s\le d_r$, and hence $a_r\le t'_s\le d_r$.
The shifted service therefore still covers every request assigned to it.

After applying these operations to all services, merge all services occurring at the same time into one service whose location is the maximum of their locations.
Assign to the merged service the union of their assigned request sets.
Its location is the largest location in this union, and its time is the smallest deadline in the union.
The merge preserves feasibility and does not increase the total cost.

Thus every feasible solution can be transformed, without increasing its cost, into a solution satisfying the stated properties.
Every service in such a solution has a time in the finite set $\{d_r:r\in J\}$ and a location in the finite set $\{x_r:r\in J\}$.
There are therefore only finitely many such solutions.
At least one is feasible, for example the solution that serves each request at its own deadline.
A minimum-cost solution among this finite family exists, and the preceding transformation shows that it is also optimal among all feasible solutions.
\end{proof}

We call a solution satisfying Lemma~\ref{lem:canonical} \emph{deadline-canonical}.
The next lemma shows how an optimal solution can be decomposed into two smaller subproblems, which leads to the recurrence used in our dynamic program.

\begin{lemma}
\label{lem:farthest-separation}
Let $J$ be an instance and $p\in\arg\max_{r\in J}x_r$.
There exists a time $t\in\{d_r:r\in J\}\cap[a_p,d_p]$ such that
\begin{equation}
\label{eq:farthest-separation}
\OPT(J)=x_p+\OPT\bigl(J^-(t)\bigr)+\OPT\bigl(J^+(t)\bigr).
\end{equation}
where $J^-(t)=\{r\in J:d_r<t\}$ and $J^+(t)=\{r\in J:a_r>t\}$.
\end{lemma}

\begin{proof}
Take a deadline-canonical optimal solution for $J$, and let $s$ be the service to which $p$ is assigned.
Write $t$ for its service time.
Since $s$ covers $p$, we have $a_p\le t\le d_p$.
Lemma~\ref{lem:canonical} also implies that $t$ is the deadline of a request assigned to $s$.
Hence $t\in\{d_r:r\in J\}\cap[a_p,d_p]$.
The location of $s$ is exactly $x_p$: it is at least $x_p$ because $s$ covers $p$, and it is at most $x_p$ because $p$ is a farthest request in $J$.

Let $C(t)=\{r\in J:a_r\le t\le d_r\}$.
Every request in $C(t)$ has location at most $x_p$, so $(t,x_p)$ covers all of them.
Reassign every request in $C(t)$ to this service, remove these requests from the assignments of all other services, and delete any service that becomes unassigned.
The resulting solution is feasible and has cost no larger than the original optimum, so it remains optimal.

Every request outside $C(t)$ lies in exactly one of $J^-(t)$ and $J^+(t)$.
Moreover, a service covering a request in $J^-(t)$ must occur strictly before $t$, whereas a service covering a request in $J^+(t)$ must occur strictly after $t$.
No service can therefore cover requests from both instances.
The remaining services split into feasible solutions $S^-$ and $S^+$ for $J^-(t)$ and $J^+(t)$, respectively.
Each solution is optimal for its instance; otherwise, replacing it by a cheaper solution would improve the optimal solution for $J$.
This proves~\eqref{eq:farthest-separation}.
\end{proof}

We now define the states of the dynamic program.
Let $\tau_1<\tau_2<\cdots<\tau_m$ be the distinct request deadlines, where $m\le n:=\lvert I\rvert$, and define the boundaries $\tau_0=-\infty$ and $\tau_{m+1}=+\infty$.
For $0\le i<j\le m+1$, define $I(i,j)=\{r\in I:a_r>\tau_i,\ d_r<\tau_j\}$ and let $F(i,j)=\OPT\bigl(I(i,j)\bigr)$.
The original instance is $I(0,m+1)=I$.
If $I(i,j)$ is nonempty, choose a farthest request $p(i,j)\in\arg\max_{r\in I(i,j)}x_r$.

Only deadlines belonging to the current state are needed as split times.
For a nonempty state, define
\begin{equation*}
K(i,j)=\left\{k:\begin{array}{l}i<k<j,\\
a_p\le\tau_k\le d_p,\\
d_r=\tau_k\text{ for some }r\in I(i,j)\end{array}\right\}.
\end{equation*}
This set is nonempty because $d_p$ is a deadline in the current state and lies in $[a_p,d_p]$.
The farthest-request separation suggests the recurrence
\begin{equation}
\label{eq:offline-recurrence}
F(i,j)=\begin{cases}0, & I(i,j)=\varnothing, \\[1.2ex]
\displaystyle x_{p(i,j)}+\min_{k\in K(i,j)}\bigl\{F(i,k)+F(k,j)\bigr\}, & I(i,j)\neq\varnothing.\end{cases}
\end{equation}

The recurrence only refers to states with smaller boundary width.
Indeed, if $i<k<j$, then both $k-i$ and $j-k$ are smaller than $j-i$.
Algorithm~\ref{alg:dp} therefore evaluates the states in increasing order of $j-i$.
It stores an optimal split and a farthest request for each nonempty state so that an optimal solution can be reconstructed by using Algorithm~\ref{alg:build}.
Here and below, $\bot$ denotes a sentinel value indicating that the corresponding variable does not hold a valid value.

\begin{algorithm}[ht]
\caption{Dynamic programming algorithm for offline LAD}
\label{alg:dp}
\begin{algorithmic}[1]
\Require Requests $I=\{(a_r,x_r,d_r)\}_{r=1}^n$
\State Sort the distinct deadlines as $\tau_1<\cdots<\tau_m$
\State Set $\tau_0\gets-\infty$ and $\tau_{m+1}\gets+\infty$
\For{$w=1$ to $m+1$}
    \For{$i=0$ to $m+1-w$}
        \State $j\gets i+w$
        \State $F[i,j]\gets0$
        \State $\texttt{split}[i,j]\gets\bot$
        \State $\texttt{pivot}[i,j]\gets\bot$
        \If{$I(i,j)\neq\varnothing$}
            \State $p\gets p(i,j)$
            \State $k^\star\gets
                \arg\min_{k\in K(i,j)}
                \{F[i,k]+F[k,j]\}$
            \State $F[i,j]\gets
                x_p+F[i,k^\star]+F[k^\star,j]$
            \State $\texttt{split}[i,j]\gets k^\star$
            \State $\texttt{pivot}[i,j]\gets p$
        \EndIf
    \EndFor
\EndFor
\State Run Algorithm~\ref{alg:build} with $(i,j)=(0,m+1)$ to obtain an optimal solution $X^\star$
\State \Return $F[0,m+1]$ and $X^\star$
\end{algorithmic}
\end{algorithm}

\begin{algorithm}[ht]
\caption{Reconstruction of an optimal offline solution}
\label{alg:build}
\begin{algorithmic}[1]
\Require Indices $i,j$
\If{$\texttt{split}[i,j]=\bot$}
    \State \Return $\varnothing$
\EndIf
\State $k\gets\texttt{split}[i,j]$
\State $p\gets\texttt{pivot}[i,j]$
\State $X^-\gets\Call{Build}{i,k}$
\State $X^+\gets\Call{Build}{k,j}$
\State \Return
    $X^-\cup\{(\tau_k,x_p)\}\cup X^+$
\end{algorithmic}
\end{algorithm}

\begin{theorem}
\label{thm:complexity}
Algorithm~\ref{alg:dp} computes an optimal offline solution for LAD in $O(n^3)$ time and $O(n^2)$ space.
\end{theorem}

\begin{proof}
We first prove the recurrence and hence the correctness of the algorithm.
The formula for an empty state is immediate.
Fix a nonempty state $I(i,j)$ and write $p=p(i,j)$.

For the upper bound, choose any $k\in K(i,j)$.
Every request in $I(i,j)$ lies in exactly one of the following three sets:
\begin{align*}
L_k&=\{r\in I(i,j):d_r<\tau_k\},\\ C_k&=\{r\in I(i,j):a_r\le\tau_k\le d_r\},\\ R_k&=\{r\in I(i,j):a_r>\tau_k\}.
\end{align*}
The definitions of the states give $L_k=I(i,k)$ and $R_k=I(k,j)$.
Since $p$ is farthest in $I(i,j)$, the service $(\tau_k,x_p)$ covers every request in $C_k$.
Combining this service with optimal solutions for $I(i,k)$ and $I(k,j)$ gives a feasible solution for $I(i,j)$.
Therefore
\begin{equation}
\label{eq:offline-upper}
F(i,j) \le x_p+ \min_{k\in K(i,j)}\{F(i,k)+F(k,j)\}.
\end{equation}

For the lower bound, apply Lemma~\ref{lem:farthest-separation} to $J=I(i,j)$.
It provides a split time $t$ that is the deadline of a request in the current state and satisfies $a_p\le t\le d_p$.
Hence $t=\tau_k$ for some $k\in K(i,j)$.
Because $p\in I(i,j)$, we also have $\tau_i<t<\tau_j$.
At this split time, $J^-(t)=I(i,k)$ and $J^+(t)=I(k,j)$.
Equation~\eqref{eq:farthest-separation} now gives
\begin{align*}
F(i,j) &=x_p+F(i,k)+F(k,j)\\ &\ge x_p+ \min_{\ell\in K(i,j)}\{F(i,\ell)+F(\ell,j)\}.
\end{align*}
Together with~\eqref{eq:offline-upper}, this proves recurrence~\eqref{eq:offline-recurrence}.
Both child states in every transition have smaller boundary width, so the bottom-up order in Algorithm~\ref{alg:dp} computes every table entry before it is used.
In particular, the returned value is $F(0,m+1)=\OPT(I)$.

We next verify the reconstructed solution.
At a nonempty state, assign to $(\tau_k,x_p)$ all requests in $C_k$.
This set contains $p$, so its largest location is exactly $x_p$.
The definition of $K(i,j)$ also guarantees a request $q\in I(i,j)$ with $d_q=\tau_k$.
Since $a_q\le d_q=\tau_k$, this request belongs to $C_k$.
Every request in $C_k$ has deadline at least $\tau_k$, and therefore $\min_{r\in C_k}d_r=\tau_k$.
Thus each reconstructed service has a nonempty assigned set, location equal to the largest assigned location, and time equal to the smallest assigned deadline.
Left descendants use only deadline indices smaller than $k$, and right descendants use only indices larger than $k$.
Hence no deadline is used twice, and the reconstructed optimal solution is deadline-canonical.

It remains to bound the resources used by the algorithm.
There are $m+2$ boundary indices and hence $\binom{m+2}{2}=O(m^2)=O(n^2)$ states.
For each state, one scan over the $n$ requests determines whether the state is empty, finds a farthest request, and identifies the deadlines present in the state.
The algorithm then examines at most $m=O(n)$ split indices.
Each transition reads two previously computed table entries and takes constant time.
Each state therefore requires $O(n)$ time, giving a total running time of $O(n^3)$.

The value, split, and pivot tables each contain $O(n^2)$ entries.
The scan of one state uses only $O(n)$ temporary space, and the reconstructed solution has at most one service per deadline.
The total space usage is therefore $O(n^2)$.
\end{proof}

\section{Experiments}
\label{sec:experiments}
In this section, we conduct numerical experiments to evaluate the performance of our proposed algorithms.
We generate data from times $1$ to $T = 100$ and locations $\{1,\ldots,N\}$ for $N=100$.
At every time--location pair $(t,x)$, the request count is sampled independently from a distribution.
We use three request-count laws following the experimental design of Bamas et al.~\cite{bamas2020primaldual}: Poisson with mean $1/N$; Lomax with shape $2$ and scale $1/N$, followed by unbiased randomized rounding; and ten-times iterated Poisson initialized at $1/N$.
Note that, for all these distributions, the expected total number of requests at each time is $1$.

We next perturb each instance with noise, and compute an optimal solution on this perturbed instance and use this as a prediction. 
For replacement rate $p$ from $0$ to $1$, we perturb each actual instance independently at every time--location pair.
With probability $p$, all actual requests at the pair are deleted; independently, with probability $p$, a fresh request count is sampled from the same distribution and new slacks are generated.
We compute an optimal solution for this perturbed instance using Algorithm~\ref{alg:dp} and use it as the raw advice.
To ensure that the advice-induced solution is feasible, we add an extra service action at the deadline and location of each request not covered by the advice.
The average competitive ratio (ACR) is defined as the average, over all trials, of the online algorithm’s total cost divided by the offline optimal cost.
A similar generation procedure has also been used in prior work~\cite{bamas2020primaldual,grigorescu2022learning}, where it is referred to as the “replacement rate” strategy.

\begin{figure*}[t]
\centering
\includegraphics[width=\textwidth]{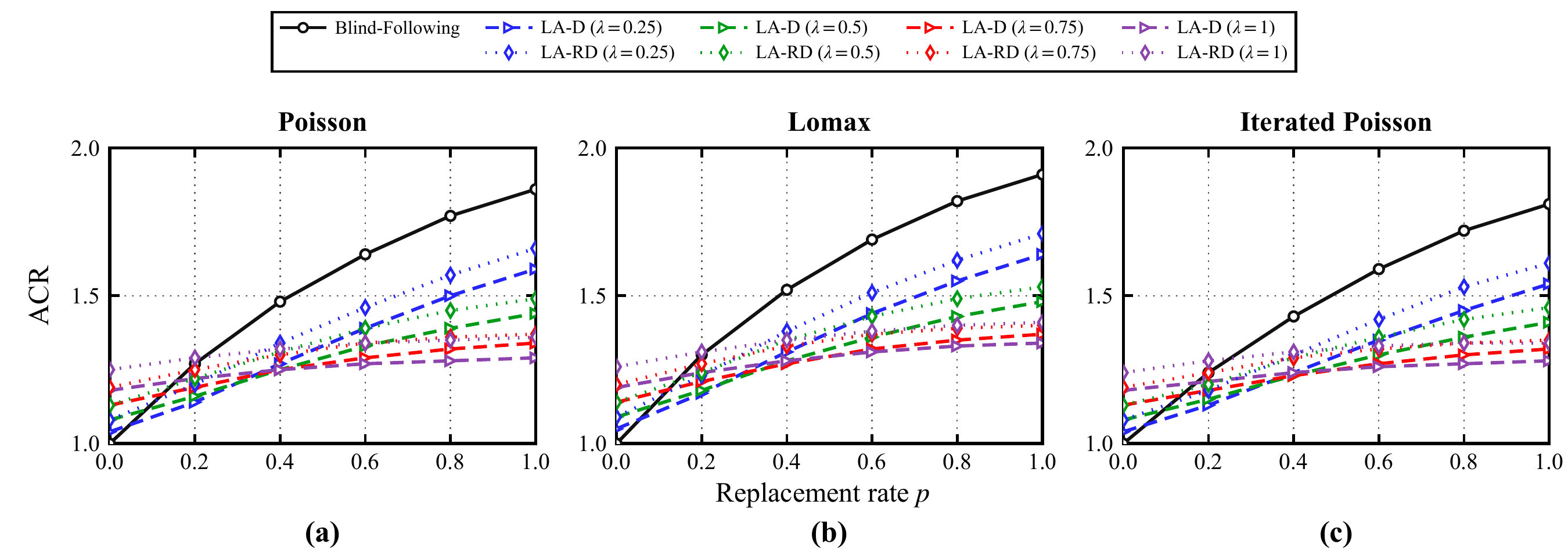}
\caption{The performance of algorithms under three distributions. (a) Poisson distribution. (b) Lomax distribution. (c) Iterated Poisson distribution.}
\label{fig:exp}
\end{figure*}

For each request distribution, we generate 100 independent random instances and evaluate Algorithm~\ref{alg:pd-doubling} (LA-D) and Algorithm~\ref{alg:learning-randomized} (LA-RD).
For the two learning-augmented algorithms, we use $\lambda\in\{0.25,0.5,0.75,1\}$.
Notably, when $\lambda=1$, our learning-augmented algorithms reduce to their respective classical online algorithms.
We also include Blind-Following as a baseline, which directly adopts the advice as a solution to the original instance.
The numerical results are shown in Fig.~\ref{fig:exp}, in which we can see the performance of the proposed algorithms.
At low replacement rates, both deterministic and randomized learning-augmented algorithms achieved lower ACRs than classical online baselines, with smaller $\lambda$ yielding better performance.
At high replacement rates, the learning-augmented algorithms exhibited higher ACRs than the classical online baselines, with smaller $\lambda$ leading to worse performance.
Moreover, despite stronger theoretical guarantees, the randomized algorithm performed worse empirically than its deterministic counterpart.

\section{Conclusions}
\label{sec:conclusion}
We considered line aggregation with deadlines under three levels of information.
With no advice, a single randomly shifted exponential grid is $e$-competitive against an oblivious adversary.
A matching lower bound shows that no randomized online algorithm can do better.
When a feasible solution is available as advice, our deterministic and randomized algorithms offer different consistency--robustness trade-offs.
With complete information, the interval structure of the line leads to an exact cubic-time dynamic program.

These results leave two main questions open.
The first is to characterize the best possible consistency--robustness trade-off.
In particular, it remains open whether a randomized algorithm can achieve near-perfect consistency while providing a substantially better robustness guarantee.
The second is to extend the tight guarantees for randomization and advice from paths to broader classes of trees.
In such an extension, the bounds should depend on structural features of the tree rather than on depth alone.

\bibliographystyle{plainnat}
\bibliography{bibliography}

\end{document}